\documentclass[twoside,a4paper,12pt]{article}
\usepackage{amsmath}
\usepackage{amsthm}
\usepackage{amssymb}
\usepackage{amscd}
\usepackage{graphicx}
\usepackage{epsfig}
\usepackage{bm}

\usepackage{latexsym}
\usepackage{amsfonts}
\usepackage{color}
\input xy
\xyoption{all}

-\headheight\advance\topmargin by -\headsep

\font\black=cmbx10 \font\sblack=cmbx7 \font\ssblack=cmbx5 \font\blackital=cmmib10  \skewchar\blackital='177
\font\sblackital=cmmib7 \skewchar\sblackital='177 \font\ssblackital=cmmib5 \skewchar\ssblackital='177
\font\sanss=cmss11 \font\ssanss=cmss8 
\font\sssanss=cmss8 scaled 600 \font\blackboard=msbm10 \font\sblackboard=msbm7 \font\ssblackboard=msbm5
\font\caligr=eusm10 \font\scaligr=eusm7 \font\sscaligr=eusm5  \font\fraktur=eufm10
\font\sfraktur=eufm7 \font\ssfraktur=eufm5

\font\bsymb=cmsy10 scaled\magstep2
\def\all#1{\setbox0=\hbox{\lower1.5pt\hbox{\bsymb
			\char"38}}\setbox1=\hbox{$_{#1}$} \box0\lower2pt\box1\;}
\def\exi#1{\setbox0=\hbox{\lower1.5pt\hbox{\bsymb \char"39}}
	\setbox1=\hbox{$_{#1}$} \box0\lower2pt\box1\;}

\def\tx#1{{\fam0\relax#1}}

\newfam\bifam
\textfont\bifam=\blackital \scriptfont\bifam=\sblackital \scriptscriptfont\bifam=\ssblackital

\newfam\blfam
\textfont\blfam=\black \scriptfont\blfam=\sblack \scriptscriptfont\blfam=\ssblack

\newfam\bbfam
\textfont\bbfam=\blackboard \scriptfont\bbfam=\sblackboard \scriptscriptfont\bbfam=\ssblackboard

\newfam\ssfam
\textfont\ssfam=\sanss \scriptfont\ssfam=\ssanss \scriptscriptfont\ssfam=\sssanss
\def\sss#1{{\fam\ssfam\relax#1}}

\newfam\clfam
\textfont\clfam=\caligr \scriptfont\clfam=\scaligr \scriptscriptfont\clfam=\sscaligr

\newfam\frfam
\textfont\frfam=\fraktur \scriptfont\frfam=\sfraktur \scriptscriptfont\frfam=\ssfraktur

\def\hpb#1{\setbox0=\hbox{${#1}$}
	\copy0 \kern-\wd0 \kern.2pt \box0}
\def\vpb#1{\setbox0=\hbox{${#1}$}
	\copy0 \kern-\wd0 \raise.08pt \box0}

\def\pmb#1{\setbox0\hbox{${#1}$} \copy0 \kern-\wd0 \kern.2pt \box0}
\def\pmbb#1{\setbox0\hbox{${#1}$} \copy0 \kern-\wd0
	\kern.2pt \copy0 \kern-\wd0 \kern.2pt \box0}
\def\pmbbb#1{\setbox0\hbox{${#1}$} \copy0 \kern-\wd0
	\kern.2pt \copy0 \kern-\wd0 \kern.2pt
	\copy0 \kern-\wd0 \kern.2pt \box0}
\def\pmxb#1{\setbox0\hbox{${#1}$} \copy0 \kern-\wd0
	\kern.2pt \copy0 \kern-\wd0 \kern.2pt
	\copy0 \kern-\wd0 \kern.2pt \copy0 \kern-\wd0 \kern.2pt \box0}
\def\pmxbb#1{\setbox0\hbox{${#1}$} \copy0 \kern-\wd0 \kern.2pt
	\copy0 \kern-\wd0 \kern.2pt
	\copy0 \kern-\wd0 \kern.2pt \copy0 \kern-\wd0 \kern.2pt
	\copy0 \kern-\wd0 \kern.2pt \box0}

\mathchardef\za="710B  
\mathchardef\zb="710C  
\mathchardef\zg="710D  
\mathchardef\zd="710E  
\mathchardef\zve="710F 
\mathchardef\zz="7110  
\mathchardef\zh="7111  
\mathchardef\zvy="7112 
\mathchardef\zi="7113  
\mathchardef\zk="7114  
\mathchardef\zl="7115  
\mathchardef\zm="7116  
\mathchardef\zn="7117  
\mathchardef\zx="7118  
\mathchardef\zp="7119  
\mathchardef\zr="711A  
\mathchardef\zs="711B  
\mathchardef\zt="711C  
\mathchardef\zu="711D  
\mathchardef\zvf="711E 
\mathchardef\zq="711F  
\mathchardef\zc="7120  
\mathchardef\zw="7121  
\mathchardef\ze="7122  
\mathchardef\zy="7123  
\mathchardef\zf="7124  
\mathchardef\zvr="7125 
\mathchardef\zvs="7126 
\mathchardef\zf="7127  
\mathchardef\zG="7000  
\mathchardef\zD="7001  
\mathchardef\zY="7002  
\mathchardef\zL="7003  
\mathchardef\zX="7004  
\mathchardef\zP="7005  
\mathchardef\zS="7006  
\mathchardef\zU="7007  
\mathchardef\zF="7008  
\mathchardef\zW="700A  

\newcommand{\be}{\begin{equation}}
\newcommand{\ee}{\end{equation}}

\newcommand{\bea}{\begin{eqnarray}}
\newcommand{\eea}{\end{eqnarray}}
\newcommand{\beas}{\begin{eqnarray*}}
	\newcommand{\eeas}{\end{eqnarray*}}
\newcommand{\Z}{{\mathbb Z}}
\def\*{{\textstyle *}}
\newcommand{\R}{{\mathbb R}}

\newcommand{\Oc}{{\mathbb O}}

\newcommand{\we}{\wedge}

\newcommand{\pa}{\partial}
\newcommand{\ti}{\times}

\def\o{\mathbf{o}}

\def\Sec{\operatorname{Sec}}

\def\sT{{\sss T}}

\def\xi{\tx{i}}

\newdir{|>}{%
	!/4.5pt/@{|}*:(1,-.2)@^{>}*:(1,+.2)@_{>}}

\def\Id{\operatorname{Id}}

\newdir{ (}{{}*!/-5pt/@^{(}}

\newcommand{\nm}[1]{\ensuremath{\Vert #1 \Vert}}

\newtheorem{theorem}{Theorem}
\newtheorem{proposition}{Proposition}
\newtheorem{lemma}{Lemma}

\theoremstyle{definition}
\newtheorem{example}{Example}
\newtheorem{definition}{Definition}

\newcommand{\lvec}[1]{\overleftarrow{#1}}
\newcommand{\rvec}[1]{\overrightarrow{#1}}

\newcommand{\e}{\mathrm{e}}

\begin{document}
	\title{\bf Discrete mechanics on unitary octonions}
	\date{}
	\author{\\ Janusz Grabowski\\ Zohreh Ravanpak
		\\ \\
		{\it Institute of Mathematics}\\
		{\it Polish Academy of Sciences}
	}
	\maketitle
	
	\begin{abstract}
		In this article we generalize the discrete Lagrangian and Hamiltonian mechanics on Lie groups to non-associative objects generalizing Lie groups (smooth loops). This shows that the associativity assumption is not crucial for mechanics and opens new perspectives. As a working example we obtain the discrete Lagrangian and Hamiltonian mechanics on unitary octonions.
	\end{abstract}


\noindent
\textbf{AMSC2010:} 17B66, 17Dxx, 20N05, 22A22, 70G45, 70Hxx

	\section{Introduction}

The main tool in the theory of geometric integrators (see e.g. \cite{MW}) is the discrete Lagrangian and Hamiltonian formalism on $G=M\ti M$, i.e. the groupoid of pairs of points of a manifold $M$. This can be generalized to arbitrary Lie groupoid, in particular any Lie group. This formalism has been studied by Weinstein \cite{weinstein} (see also \cite{MMM}).	

	In \cite{MV}, Moser and Veselov  consider the Lagrangian and Hamiltonian formalisms for discrete mechanics on a Lie group. The Lagrangian function
$L$ is defined on a Lie group $G$, and the dynamical system is given by a diffeomorphism from $G$ to itself. The corresponding Hamiltonian system is the mapping from the dual Lie algebra $\mathfrak g^* $ to itself for which $L$ is the generating function.

Lie groupoids have been recently used for a geometric formulation of the Lagrangian formalism and information geometry in many papers \cite{FZ, GGKM,IMMM,IMMP,MMM,MMS,Stern,weinstein}. Infinitesimal parts of Lie groupoids are Lie algebroids and mechanics on Lie algebroids has been also extensively studied \cite{IMMS, LMM, Li, Martinez, weinstein}. More general algebroids in which the Jacobi identity is not satisfied (\cite{GU1,GU2}) have been used in this context as well \cite{GG,GG1,GGU,GJ,GLMM}. The idea is based on the concept of a \emph{Tulczyjew triple} \cite{Tul1,Tul2}. For general theory of Lie groupoid and Lie algebroids we refer to Mackenzie \cite{Ma}.

Let $G\rightrightarrows M$ be a Lie groupoid, $\alpha, \beta : G \to M$ being its source and target maps, with a multiplication map $m : G^{(2)} \to G$, where $G^{(2)} = \{(g,h) \in G\times G| \; \beta(g) = \alpha(h)\}$. Denote its corresponding Lie algebroid by $AG$ represented by the normal bundle $\zn(M)=\sT G_{|M}/\sT M$ to the submanifold of units $M\subset G$. Sections of $AG$ are represented by the left-invariant $\lvec X$ (or right-invariant $\rvec X$) vector fields on $G$ associated with $X\in\Sec(AG)$.

Lagrangian mechanics on a groupoid $G$ for a smooth, real-valued function $L$ on $G$ is defined as follows \cite{weinstein}.

Let $L^{(2)}$ be the restriction to the set of composable pairs $G^{(2)}$ of the function $(g,h) \to L(g) + L(h)$ and $\Sigma_L \subset G^{(2)}$ be the set of critical points of $L_{(2)}$ along the fibers of the multiplication map $m$; that is, the points in $\Sigma_L$ are stationary points of the function $L(g)+L(h)$ when $g$ and $h$ are restricted to admissible pairs with the constraint that the product $gh$ is fixed. Variations of the constraint are of the form $(gu,u^{-1}h)\in G^{(2)}$.

A solution of the Hamilton principle for the Lagrangian function $L$ is a sequence $...,g_2,g_1,g_0,g_1,g_2,...$ of elements of $G$, defined on some “interval” in $\mathbb Z$, such that $(g_i,g_{i+1}) \in \Sigma_L$ for each $i$. The Hamiltonian formalism for discrete Lagrangian systems is based on the fact that each Lagrangian submanifold of a symplectic groupoid (see \cite{We}) determines a Poisson automorphism on the base Poisson manifold. Recall that the cotangent bundle $\sT ^*G$ is, in addition to being a symplectic manifold, a Lie groupoid itself, the base being $A^*G$; notice that both manifolds are naturally Poisson. The source and target mappings $\tilde\alpha , \tilde\beta : \sT^* G \to A^*G$ are Poisson maps induced by $\alpha$ and $\beta$. In \cite{MMS}, the authors showed that Lagrangian submanifolds of symplectic groupoids give rise to discrete dynamical system.

Discrete Euler-Lagrange equations on Lie groupoids can be derived from the variational principles. In \cite{MMM},  the discrete Euler-Lagrange equations take the form
\be\label{e1}
\lvec X(g_i)(L)-\rvec X(g_{i+1}) (L)=0
\ee
on a Lie groupoid $G\rightrightarrows M$, for every section $X$ of $AG$. Note that $(g_i,g_{i+1})\in G^{(2)}$ and the left and right arrow denotes the right and left-invariant vector field on $G$ associated with $X\in \Sec(AG)$ understood as a section on the normal bundle.

A development of discrete Lagrangian mechanics on a Lie groups and groupoids has been developed in many papers (e.g. \cite{FZ,IMMM,IMMP,MMM, MP,MR, Stern,weinstein}).
Nevertheless, the generalization of the discrete mechanics to non-associative objects
is still lacking, and the aim of this paper is to fill this gap by presenting a systematic approach
for the construction of discrete Lagrangian and Hamiltonian formalism on smooth loops \cite{Br1}.

The theory of smooth quasi-groups and loops has already
	started to find interesting applications in geometry and physics.
	The remarkable development of smooth quasigroups and loops theory since the pioneering
	works of Mal'cev in 1955 (see \cite{mal}) was presented by Lev V. Sabinin \cite{Sab}, where the large bibliography on the subject
	is given. We refer also to the books \cite{Bel,Pfl} and the survey articles \cite{Sab1,Smi} if terms and
	concepts from non-associative algebra are concerned.

As a working example we will develop the discrete Lagrangian and Hamiltonian formalism on unitary octonions $\Oc_1$ (understood as an inverse loop in the algebra of octonions $\Oc$ or a subloop in the loop $\Oc^\ti$ of invertible octonions) which as a manifold is the seven-sphere.

It is well known that $S^0$, $S^1$, $S^3$ and $S^7$ are only spheres which are parallelizable and they correspond to elements of unit norm in the normed division algebras of the real numbers, complex numbers, quaternions, and octonions. The first three spheres are Lie groups ($S^0=\Z_2$, $S^1=U(1)$, $S^3=SU(2)$), but $S^7$ is the only parallelizable sphere which is not a Lie group since it is not associative.	

The left and right translations in the loop $\Oc_1$ act as diffeomorphism, so the left and right prolongations $\lvec X$ and $\rvec X$  of $X$ in the tangent space $\o_1=\sT_{e_0}\Oc_1=A\Oc_1$ at the neutral element $e_0\in\Oc_1$ to the loop $\Oc_1$ are well defined vector fields and the Euler-Lagrange equation (\ref{e1}) makes sense also in this case, although the variational approach is not applicable. Note that because of the lack of associativity the vector fields $\lvec X$ and $\rvec X$ are no longer left  nor right invariant, so that the tangent algebra $\o_1$ is not a Lie algebra. It is enough to observe that the other concepts like the Legendre map, Hamiltonian map, etc. are in the case of a Lie group built on prolongations $\lvec X$ and $\rvec X$, so versions of  Lagrangian and Hamiltonian formalisms can be formulated also for smooth loops. To develop such versions is the main aim of this paper.

	\section{Smooth loops}
Let us recall that a \emph{loop} is an algebraic structure $<G,\cdot, e>$ with a binary operation (written usually as juxtaposition, $a\cdot b=ab$) such
	that
	$r_a: x \mapsto xa$ (the \emph{right translation}) and
	$l_a: x \mapsto ax$ (the \emph{left translation})
	are permutations of $G$, equivalently, in which the equations
	$ya = b$ and $ax = b$ are uniquely solvable for $x$ and $y$ respectively, with a two-sided identity element,
	$e$, $e x=x e=x$. A loop $<G,\cdot , e>$ with identity $e$ is called an \emph{inverse
		loop} if it is equipped with a smooth inversion map $\iota : G \to G$ which we denote simply by $\iota (a) = a^{-1}$. In other words,
 to each element $a$ in G there corresponds an element $a^{-1}$ in $G$ such that
	$$a^{-1}(a b) = (b a) a^{-1} =b$$ for all $b\in G$.
	It can be then easily shown that in an inverse loop $<G, \cdot, {}^{-1}, e>$ we have,
	for all $a, b \in G$,
	\be\label{inv}aa^{-1} = a^{-1}a = e,\quad (a^{-1})^{-1} = a,\quad \text{and}\quad  (ab)^{-1} = b^{-1} a^{-1}\,.\ee
The above identities imply that
\be\label{inv1} \iota(e)=e\,,\quad\iota^2=\Id_G,\quad \text{and}\quad \zi\circ l_a=r_{a^{-1}}\circ\zi\,.
\ee
	Loops having only one inverse are called \emph{left inverse loops} (resp. \emph{right inverse loops}). Left inverse loops, appear naturally as algebraic structures
	on \emph{transversals} or \emph{sections} of a subgroup in a group. In this case the homogeneous structures are equipped with a binary operation. This observation, going
	back to R. Baer \cite{Bae} (cf. also \cite{Fo,KW}), lies at the heart of much current research on loops, also in differential geometry and analysis.

A \emph{smooth loop} $ G$ is a smooth manifold equipped with a smooth multiplication, $m :G \times  G \to  G$, $(g, h) \longmapsto m(g, h) = gh$, such that the left and right translations are diffeomorphisms, together with an identity element $e\in G$ such that $eg=ge=g$ for every $g\in G$. Let $\mathfrak g=\sT_eG$. If $G$ is a smooth loop with the smooth inverse $\zi:G\to G$ (smooth inverse loop), then
\be\label{inv12}\zi_*(X)=-X\quad \text{for}\quad X\in\mathfrak g\,.
\ee
 Indeed, if $\zg:\R\to G$ is a curve in $G$ such that $\zg(0)=e$ and $\zg$ represents $X\in\mathfrak g$, then
$$\frac{d}{dt}|_{t=0}\left(\zg(t)(\zg(t))^{-1}\right)=X+\zi_*(X)=0\,.$$

Let $X_e\in\mathfrak g$ be a vector field in $\sT_eG$. We can left-translate (resp. right-translate) the value of $X_e$ by the tangent of the left (resp. right) translation by $g$. However, although they are not invariant vector fields anymore due to the lack of associativity, we still are able to define the left (resp. right) \emph{prolongations} of $X_e$  to vector fields $\lvec X$ (resp. $\rvec X$) on $G$ using the tangent maps at $g \in G$ to the left translation $l_g$ and right translation $r_g$:
	\[
	\lvec X_g=D_e(l_g)( X_e),\quad  \rvec X_g=D_e(r_g) (X_e).
	\]
Here, $D$ denotes the derivative. This prolongations are smooth vector fields, because for a smooth function $f$ defined on $ G$ and a smooth curve $\gamma$ such that $\gamma(0)=e$ and $\frac{d}{dt}|_{t=0}\gamma(t)=X_e$,  we get the following smooth function
		\[
	\begin{array}{rcl}
	\lvec Xf(g)&=&D_e(l_g) (X_e)f=(X_e)(f\circ l_g)=\frac{d}{dt}|_{t=0}(f\circ l_g \circ \gamma(t))\\[6pt]
	&=& \frac{d}{dt}|_{t=0}(f(g\gamma(t)))=\frac{d}{dt}|_{t=0}(f\circ m(g,\gamma(t)))\\
	\end{array}
	\]
and similarly for $\rvec Xf(g)$. According to (\ref{inv1}) and (\ref{inv12}), we have
\be\label{inv2}
\zi_*(\lvec X)=-\rvec X\circ\zi\,,
\ee
but due to non-associativity we cannot infer that there is $[X,Y]\in\mathfrak g$ such that $[\lvec X,\lvec Y]=\lvec{[X,Y]}$ nor $[\rvec X,\rvec Y]=\rvec{[X,Y]}$. Moreover, in general, we do not have $[\lvec X,\rvec Y]=0$ nor $[\rvec X,\rvec Y]=-[\lvec X,\lvec Y]$.

However, the tangent space at the identity $\sT_e G\cong\mathfrak g$ inherits a skew-symmetric bilinear product $[ \cdot,\cdot ]_l$ from the Lie product of the left prolongations of vector fields over the loop. In other words, $[X,Y]_l=[\lvec X,\lvec Y]_e$. This is indeed a bilinear product, since for $a\in\R$ we have $\lvec{aX}=a\lvec X$. The Jacobi identity does not hold due to the non-associativity. So, $\mathfrak g$ it is not a Lie algebra but a  \emph{skew-algebra}, that is, a vector space equipped with a skew-symmetric binary operation. A similar bracket $[ \cdot,\cdot ]_r$ we obtain from the right prolongations, but in general they do not differ only by sign.
The skew-algebra structure on $\mathfrak g$ corresponds to a linear \emph{Leibniz structure} on $\mathfrak g^*$, i.e. a linear bivector field, exactly like a Lie algebra structure corresponds to a linear Poisson structure on the dual space.

In the paper \cite{mal} a local \emph{diassociative} analytic loop $G$ was considered.
Diassociativity means that any two elements generate a genuine subgroup.
Since the multiplication in a loop is a binary operation and the loop is diassociative,
one may write the analogue of the Campbell-Hausdorff series, which depends
only on one skew-symmetric bilinear operation $[\cdot,\cdot]$ in the tangent
space $\mathfrak g=\sT_eG$. In our notation this bracket coincides with $[\cdot,\cdot]_l$. This algebra is a
\emph{binary-Lie algebra}, that is, any two of its elements generate a subalgebra which
is a Lie algebra. Any binary-Lie algebra generates, by means of the Campbell-
Hausdorff formula, a diassociative local loop. Thus we get a diassociative smooth
loops -- binary-Lie algebras theory generalizing the Lie groups -- Lie algebras
theory.

In the same paper of Mal’cev smooth Moufang loops were also considered, as particular cases of diassociative loops.		
	A loop is called a \emph {Moufang loop} if it satisfies any of the three following equivalent conditions \cite{JM}
	\[
	((ax)a)y = a(x(ay)),\quad	((xa)y)a = x(a(ya)), \quad(ax)(ya)= (a(xy))a.
	\]
	The tangent algebra of a smooth Moufang loop is a Mal'cev algebra which is a binary-Lie algebra satisfying \cite{nagy}
	\[
	\left[ [X,Y],[X,Z]\right] =\left[ \left[ [X,Y],Z\right] ,X\right] +\left[ \left[ [Y,Z],X\right] ,X\right]+\left[ \left[ [Z,X],X\right] ,Y\right],
	\]	
	for every $X,Y,Z$. There is  again a sort of Lie's Third Theorem for smooth Moufang loops and Mal'cev algebras \cite{kuz,mal,nagy}.

The following is well known.
\begin{theorem}\label{t1}
Invertible octonions $\Oc^\ti$ form a smooth inverse Moufang loop under the octonion multiplication..
\end{theorem}
It is also well known that the tangent and cotangent bundles of a Lie group are Lie groupoids themselves, the tangent bundle $\sT G$ is a Lie group with the unit $(e,0)$ and the cotangent bundle $\sT^*G$ is a smooth groupoid over $\mathfrak g^*=AG$ \cite{Ma,We}.
The multiplication relation in the first case is the tangent bundle $\sT m\subset \sT G\ti\sT G\ti\sT G$ of the multiplication relation $m\subset G\ti G\ti G$ in $G$ and for $\sT^* G$ it is the annihilator  $(\sT m)^0\subset\sT^* G\ti\sT^*G\ti\sT^*G$, where the pairing with the third $\sT^*G$  we take with the minus sign. This can be easily extended to smooth loops.
\begin{theorem}\label{t2}
	The tangent bundle $\sT G$ of a smooth loop $G$ is a smooth loop under the multiplication
\[
D_{(g,h)}m(X_g,Y_h)=D_g(r_h)(X_g)+D_h(l_g)(Y_h),
\]
for $X_g\in \sT_gG$ and $Y_h\in \sT_hG$.
\end{theorem}
However, as we will see later, for the cotangent bundle of a smooth loop the multiplication $(\sT m)^0\subset\sT^* G\ti\sT^*G\ti\sT^*G$ gives in general not a smooth loop.
	
	\section{Discrete mechanics on smooth loops}
\subsection{Lagrangian formalism}	
	 Let us first recall the discrete Lagrangian mechanics on Lie groups. A discrete Lagrangian system consists of a Lie group $G$ and a smooth, real-valued function $L$ on $G$. We define a function $(g,h) \to L(g) + L(h)$ of elements $g,h\in G$. A solution of the Lagrange equations for the Lagrangian function $L$ is a sequence $g_0,g_1,g_2,...$ of elements $G$ such that $(g_i,g_{i+1})$ are the stationary points of the function $L(g_i)+L(g_{i+1})$ for every $i$. 

	Discrete Lagrangian systems on Lie groups can be based on variational principles as follows.	
	The variational principle for a Lie group $G$ with Lie algebra $\mathfrak g$ is based on a set of sequences
	\[
	\mathcal C^N_g=\{(g_1,g_2,...,g_N)\in G^N\mid g_1g_2\cdots g_N=g\in G\}\,.
	\]

	Take a tangent vector at $(g_1,g_2...,g_N)$ which can be understood as the tangent vector of a curve $c(t) \in \mathcal C^N_g$ passing through $(g_1,g_2...,g_N)$ at $ t=0 $.
It is easy to see the following.
\begin{lemma}\label{l1} In a Lie group $G$ we have $g_ig_{i+1}=g'_ig'_{i+1}$ if and only if there is $h\in G$ such that $g'_i=g_ih$ and $g'_{i+1}=h^{-1}g_{i+1}$.
\end{lemma}
By Lemma \ref{l1} the curve $c(t)$ is necessarily of the form
	\begin{equation}\label{vari}
	 c(t)=(g_1\gamma_1(t),(\gamma_1(t))^{-1}g_2\gamma_2(t),...,(\gamma_{N-2}(t))^{-1}g_{N-1}\gamma_{N-1}(t),(\gamma_{N-1}(t))^{-1}g_N ),
	\end{equation}
	such that $\gamma_i:t\in(-\epsilon,\epsilon)\subseteq \mathbb R \to G$
	are the integral curves of the left invariant vector field corresponding to $X_i\in  T_e G$ that passes through the identity, that is $\gamma_i(0)=e$.	
	
	Therefore the tangent space of $\mathcal C^N_g$ at $(g_1,\dots,g_N)$ can be identified with
	\[
	\sT_{g_1,\dots,g_N}\mathcal C^N_g=\{(X_1,X_2,...,X_{N-1})\in \mathfrak g^{N-1}\mid X_i\in T_{e}G\cong \mathfrak g\}.
	\]
	The curve $c$ is called a \emph{variation} of $(g_1,g_2...,g_N)$ and $(X_1,X_2...,X_{N-1})$ is called \emph{infinitesimal variational} of $(g_1,g_2...,g_N)$. Now, we define the discrete action sum associated to the Lagrangian $L$
	\[
	\mathbb S L=\sum_{k=1}^NL(g_k).
	\]
	According to the Hamilton's principle of critical action, the sequence $(g_1,g_2...,g_N)$ is a solution of the Lagrangian system if and only if $(g_1,g_2...,g_N)$ is a critical point of $\mathbb S L$. Therefore, we calculate
	\[
	\frac{d}{dt}\mid_{t=0}\mathbb SL(c(t))=
	\sum_{k=1}^{N-1} \left[ \lvec X_k(g_k)(L)-\rvec X_k(g_{k+1}) (L)\right] =0,
	\]
	where $X_k\in\mathfrak g$. These equations are called to be \emph{discrete Euler-Lagrange equations}.

		In the category of smooth loops, because of the lack of associativity there is no variant of Lemma \ref{l1} so not clear variations like (\ref{vari}). But still we can define the discrete Euler-Lagrange equations using the smooth prolongation of vector fields as follows.
		
	\begin{definition}
		The {\emph discrete Euler-Lagrange equations} for a discrete Lagrangian system on a smooth loop $G$ with Lagrangian $L:G\to \mathbb R$ is given by equations
		\begin{equation}\label{ELE.}
		\lvec X(L)(g_i)-\rvec X (L)(g_{i+1})=0
		\end{equation}
		for every $X\in T_{\e} G$, where $\lvec X$ and $\rvec X$ are the left and right prolongation, respectively.	
	
		A sequence $g_1,g_2,...$ of elements $G$ is a solution of the Euler-Lagrange equations if $g_i$ and $g_{i+1}$ satisfy (\ref{ELE.})  for $i=1,2,\dots$ .
	\end{definition}
	Let $\gamma_L:G\to G$ be a smooth map on a smooth loop $G$ for which the couples $(g,\gamma (g))$ are solutions of Euler-Lagrange equations for $L$. The map $\gamma_L:G\to G$ is called a \emph{discrete flow} or \emph{discrete Lagrangian evolution operator}  for $L$.	
	
	We have the discrete Legendre transformation for smooth loops similar to what we have for Lie groups \cite{MP}.	
Given a Lagrangian $L:G\to \mathbb R$ on smooth loop $G$ with the skew-algebra $\mathfrak g$, two \emph{discrete Legendre transformations} $\mathbb F^+L=l^*_g\circ dL:G\to \mathfrak g^*$ and $\mathbb F^-L=r^*_g\circ dL:G\to \mathfrak g^*$, where $dL:G\to \sT ^*G$, are as follows
	\[
\mathbb F^+L(g)(X)=\lvec X (L)(g),\quad \mathbb F^-L(g)(X)=\rvec X (L)(g),
	\]
	for $X\in \mathfrak g$. Clearly, $l^*$ and $r^*$ are the pull backs of left and right translations.
Directly from the definitions we get the following
\begin{proposition}\label{p2}
$\gamma_L:G\to G$ is the discrete flow for the Lagrangian $L:G\to\R$ if and only if
\be\label{df}
\mathbb F^-L\circ\zg_L=\mathbb F^+L\,.
\ee
\end{proposition}
\begin{theorem}\label{tl}
For an inverse smooth loop $G$ the Legendre map $\mathbb F^+L$ is regular at $g$ if and only if $\mathbb F^-L$ is regular at $g^{-1}$.
Moreover, $\mathbb F^+L$ is a diffeomorphism if and only if $\mathbb F^-L$ is a diffeomorphism.
\end{theorem}
\begin{proof}
According to (\ref{inv2}), $\zi_*\rvec X(g)=\lvec X(g^{-1})$ and $\zi$ is a diffeomorphism, so they are simultaneously local diffeomorphisms and injective.
\end{proof}

\begin{definition}
A discrete Lagrangian $L:G\to \mathbb R$ on smooth loop $G$ is said to be \emph{regular} if and only if the Legendre transformation $\mathbb F^+L$ is a local diffeomorphism.	If $\mathbb F^+L$  is global diffeomorphism, $L$ is called to be \emph{hyperregular}.
\end{definition}
\begin{theorem}\label{tl1} For an inverse smooth loop $G$ the following are equivalent:
\begin{itemize}
\item A discrete Lagrangian $L:G\to \mathbb R$ on smooth loop $G$ is to be {regular};
\item $\mathbb F^-L$ is a local diffeomorphism;
\end{itemize}
Moreover, $L$ is hyperregular if and only if $\mathbb F^- L$ is a global diffeomorphism. In this case the discrete Lagrangian evolution operator is a diffeomorphism.
\end{theorem}
\begin{proof}
This follows directly from Theorem \ref{tl} and Proposition \ref{p2}.
\end{proof}

\subsection{Hamiltonian formalism}
In the category of Lie group, the cotangent bundle of a Lie group is a symplectic groupoid over the dual of the tangent algebra.  The Hamiltonian formalism for discrete Lagrangian systems on the Lie group $G$ with Lagrangian $L:G\to \mathbb R$ is based on the fact that $L$ generates a Lagrangian submanifod $dL(G)\subset \sT^*G$ of the cotangent groupoid which, under a hypothesis of non-degeneracy, determines a Poisson map from $\mathfrak g^*$ to itself. In this case, if $g$ and $h$ are solutions of the discrete Euler-Lagrange equations for the regular Lagrangian $L$, (see \cite{MMM}) then there exist two open subsets $U_g$ and $U_h$ of $G$ and a \emph{discrete Lagrangian evolution operator} $\gamma_L:U_{g} \to U_{h}$ such that $\gamma_L(g')=h'$ whenever $(g',h')$ satisfy the corresponding Euler-Lagrange equations, and $\gamma$ is a unique such diffeomorphism. If $L$ is hyperregular, then $\gamma_L=(\mathbb F^-L)^{-1} \circ\mathbb F^+L$. For a hyperregular Lagrangian function $L:G\to \mathbb R$, pushing forward to $\mathfrak g^*$ with the discrete Legendre transformations gives the \emph{discrete Hamiltonian evolution operator} $\tilde \gamma_L:\mathfrak g^*\to  \mathfrak g^*$ given by
	\[
	\tilde \gamma L=\mathbb F^+ L \circ (\mathbb F^-L)^{-1}\,.
	\]
	
Let now $G$ be a smooth loop with the skew-algebra $\mathfrak g$ and the dual $\mathfrak g^*$. There are two projections $\alpha, \beta: \sT ^*G\to \mathfrak g^*$ such that
	\begin{equation}\label{eq:cotangent:groupoid}
\begin{array}{rcl}
\left\langle\beta(\mu _g),X\right\rangle &=& \left\langle\mu _g,D_e (l_g )(X)\right\rangle ,
\mbox{ for }\mu _g\in\sT^*_ g G \mbox{ and } X\in \mathfrak g,\\
\left\langle  \alpha (\nu_h),Y\right\rangle &=&\left\langle \nu _h ,D_e( r_h) (Y)\right\rangle,
\mbox{ for }\nu _h\in \sT^* _hG \mbox{ and }Y\in \mathfrak g\,.
\end{array}
\end{equation}
In other words,
\be\label{e3}
\left\langle\beta(\mu _g),X\right\rangle=\left\langle\mu_g,\lvec X(g)\right\rangle\,,\qquad
\left\langle\alpha(\nu _h),X\right\rangle=\left\langle\nu_h,\rvec X(h)\right\rangle\,.
\ee
Now, we consider a discrete Lagrangian function $L:G\to \mathbb R$ on the smooth loop $G$. The cotangent bundle of $\sT ^*G$ is equipped with a canonical symplectic structure but the lack of associativity is an obstacle for defining a natural loop structure on $\sT^*G$ analogous to the Lie group. In other words in general there is no natural partial multiplication ('loopoid structure') on $\sT^*G$. Setting aside the 'loopoid structure', for any function $L:G\to \mathbb R$ on manifold $G$ the submanifolds $dL(G)\subset \sT ^*G$ is a Lagrangian submanifold of the cotangent bundle. The discrete Euler-Lagrange dynamics can be equivalently described as follows.	
\begin{definition}
 Let $G$ be a smooth loop and $L$ a discrete Lagrangian function on it. A sequence $ \mu_1,...,\mu_n\in \sT^*G$ satisfies the \emph{discrete Lagrangian dynamics} if $\mu_1,...,\mu_n\in dL(G)$ and they are composable sequence in $\sT ^*G$, that is
  \[
   \beta(\mu_k)= \alpha(\mu_{k+1}),\quad k=1,...,n-1\,.
  \]
\end{definition}

\begin{theorem}
 Let $G$ be a smooth loop equipped with a discrete Lagrangian $L:G\to \mathbb R$. Then a sequence $ \mu_1,...,\mu_n\in \sT^*G$ satisfies the discrete Lagrangian dynamics of $dL(G)\subset\sT ^*G$ if and only if
 \[
 \mu_k=dL(g_k)\quad \mbox{for some} \quad g_k\in G,\quad k=1,\dots,n,
 \]
 and the discrete Euler -Lagrangian equations $	\lvec X(L)(g_k)=\rvec X(L) (g_{k+1})$ are satisfied, $k=1,\dots,n-1$.
\end{theorem}

\begin{proof}
	It is enough to consider the discrete Legendre transforms of $L$ as $\mathbb F^+L= \beta \circ dL,$ $\mathbb F^-L= \alpha \circ dL:G\to \mathfrak g^*$. For more details confront \cite{MMS}.	
	
	\end{proof}

  We see that the sequence $\mu_1,...,\mu_n\in \sT^*G$ of composable pairs satisfies the discrete Lagrangian dynamics if and only if we have the relation $(\mu_k,\mu_{k+1})\in dL(G)\times dL(G)$, for each pairs of successive elements.
	
	Now, if the restricted map $\alpha: dL(G)\to \mathfrak g^*$ is a (local) diffeomorphism, then the relation $(\mu_k,\mu_{k+1})\in dL(G)\times dL(G)$ is the graph of an explicit flow map $\mu_k\to \mu_{k+1}$ given by the composition $(\alpha\mid_{dL(G)} )^{-1}\circ \beta\mid_{dL(G)}$. If the restricted map $\beta: dL(G)\to \mathfrak g^*$ is also a local diffeomorphism, then the flow is locally reversible with the inverse $(\beta\mid_{dL(G)} )^{-1}\circ \alpha\mid_{dL(G)}$. When both the restricted maps $\alpha$ and $\beta$ are diffeomorphism we say $dL(G)$ is a \emph{Lagrangian bisection} because it is simultaneously a section of $\za$ and $\zb$. Obviously, the restricted maps $\alpha$ and $\beta$ correspond precisely to the discrete Legendre transforms $\mathbb F^{\pm}$ and the local bisection condition corresponds to regularity of discrete Lagrangian $L$. Therefore, if $L$ is regular, then the discrete Lagrangian flow map $(\alpha\mid_{dL(G)} )^{-1}\circ \beta\mid_{dL(G)}$ is a local diffeomorphism on $dL(G)$, then the \emph{discrete Hamiltonian flow map} is given by reversing the order of composition $\beta\mid_{dL(G)} \circ (\alpha\mid_{dL(G)})^{-1}$ which is a local diffeomorphism on $\mathfrak g^*$.
	
	\section{The smooth loop of octonions}
	The octonions $\mathbb{O}$ are the noncommutative non-associative algebra which is one of the four division algebras that exist over the real numbers. The most elementary way to construct the octonions is to give their multiplication table. Every octonion can be expressed in terms of a natural basis $\{e_0,e_1,\cdot,\cdot,\cdot,e_7\}$ where $e_0=1$ represents the identity element and the imaginary octonion units $e_i$, $\{i=1,...,7\}$ satisfy the multiplication rule
	$e_i e_j = −\delta_{i}^{j} +f_{ijk}e_k$, where $\delta_{i}^{j} $ is the Kronecker's delta and $f_{ijk}$'s are completely anti-symmetric structure constants which read as \cite{Sc}
	\[
	f_{123} =  f_{147} =  f_{165} =  f_{246 }=  f_{257} =  f_{354} =  f_{367} = 1.
	\]
	
	The multiplication is subject to the relations
	\[
	\forall i\ne 0\quad [e_i^2=-1]\,, \qquad
	e_ie_j=-e_je_i,\quad \mbox{for} \quad i\ne j\ne 0.
	\]
	
and the following multiplication table.

	\begin{center}
		{\footnotesize   } {\footnotesize  Multiplication table
		}\\
	\end{center}
	\begin{center}
		\begin{tabular}{| l || l | l | l | l | l | l | l | l | l | p{15mm} }
			
			\hline\hline
			\vspace{-1mm}

			{\scriptsize $e_ie_j$}&{\scriptsize $~e_0$}
			& {\scriptsize $~e_1$ }& {\scriptsize $~e_2$ }& {\scriptsize $~e_3$ }& {\scriptsize $~e_4$ }& {\scriptsize $~e_5$ }& {\scriptsize $~e_6$ }& {\scriptsize $~e_7$ }  \smallskip\\
			\hline\hline
			
			{\scriptsize $~e_0$ }&{\scriptsize $~e_0$}
			& {\scriptsize $~e_1$ }& {\scriptsize $~e_2$ }& {\scriptsize $~e_3$ }& {\scriptsize $~e_4$ }& {\scriptsize $~e_5$ }& {\scriptsize $~e_6$ }& {\scriptsize $~e_7$ } \\
			\hline
			{\scriptsize $~e_1$ }&{\scriptsize $~e_1$}
			& {\scriptsize $-e_0$ }& {\scriptsize $~e_3$ }& {\scriptsize $-e_2$ }& {\scriptsize $~e_5$ }& {\scriptsize $-e_4$ }& {\scriptsize $-e_7$ }& {\scriptsize $~e_6$ } \\
			\hline
			{\scriptsize $~e_2$}&{\scriptsize $~e_2$}
			& {\scriptsize $-e_3$ }& {\scriptsize $-e_0$ }& {\scriptsize $~e_1$ }& {\scriptsize $~e_6$ }& {\scriptsize $~e_7$ }& {\scriptsize $-e_4$ }& {\scriptsize $-e_5$ }  \\
			\hline
			{\scriptsize $~e_3$}&{\scriptsize $~e_3$}
			& {\scriptsize $~e_2$ }& {\scriptsize $-e_1$ }& {\scriptsize $-e_0$ }& {\scriptsize $~e_7$ }& {\scriptsize $-e_6$ }& {\scriptsize $~e_5$ }& {\scriptsize $-e_4$ }  \\
			\hline
			{\scriptsize $~e_4$ }&{\scriptsize $~e_4$}
			& {\scriptsize $-e_5$ }& {\scriptsize $-e_6$ }& {\scriptsize $-e_7$ }& {\scriptsize $-e_0$ }& {\scriptsize $~e_1$ }& {\scriptsize $~e_2$ }& {\scriptsize $~e_3$ }  \\
			\hline
			{\scriptsize$~e_5$ }&{\scriptsize $e_5$}
			& {\scriptsize $~e_4$ }& {\scriptsize $-e_7$ }& {\scriptsize $~e_6$ }& {\scriptsize $-e_1$ }& {\scriptsize $-e_0$ }& {\scriptsize $-e_3$ }& {\scriptsize $~e_2$ }  \\
			\hline
			{\scriptsize$~e_6$ }&{\scriptsize $~e_6$}
			& {\scriptsize $~e_7$ }& {\scriptsize $~e_4$ }& {\scriptsize $-e_5$ }& {\scriptsize $-e_2$ }& {\scriptsize $~e_3$ }& {\scriptsize $-e_0$ }& {\scriptsize $-e_1$ }  \\
			\hline
			{\scriptsize$~e_7$ }&{\scriptsize $~e_7$}
			& {\scriptsize $-e_6$ }& {\scriptsize $~e_5$ }& {\scriptsize $~e_4$ }& {\scriptsize $-e_3$ }& {\scriptsize $-e_2$ }& {\scriptsize $~e_1$ }& {\scriptsize $-e_0$ }  \\
			\hline
			
		\end{tabular}
	\end{center}

\bigskip
	The associator $[g,h,k] =(gh)k - g(hk)$ of three octonions does not vanish in general but octonions satisfy a weak form of associativity known as alternativity, namely $[g,h,g]=0$. The reason is that, two octonions determine a quaternionic subalgebra of the octonions, so that any product containing only two octonionic directions is associative (diassociativity).
	
	The octonions are a generalization of the complex numbers, with seven imaginary units, so octonionic conjugation is given by reversing the sign of the imaginary basis units. Conjugation is an involution of $\mathbb{O}$  satisfying $(gh)^* = h^* g^*$. The inner product on $\mathbb{O}$ is inherited from $\mathbb R^8$ and can be rewritten
	\begin{equation}\label{inner}
	\left\langle g,h\right\rangle = \frac{(gh^* + hg^*)}{2} = \frac{(h^*g + g^*h)}{2}\in\R\,,
	\end{equation}
	and the norm of an octonion is just $\|g \| ^2 = gg^*$ which satisfies the defining property of a normed division algebra, namely $\|gh\| = \| g\| \| h \|$.
The scalar product is invariant with respect to the multiplication: $\langle ag,ah\rangle=\langle g,h\rangle$ for $a\ne 0$.

Every nonzero octonion  $g\in \mathbb{O}$ has an inverse $g^{-1}=\frac{g^*}{\| g \| ^2}$, such that
	\begin{equation}\label{inverse}
gg^{-1}=g^{-1}g=1,	
	\end{equation}
	which makes the set of invertible octonions to be an inverse loop with respect to the octonion multiplication.
	 We remark that the inverse is a genuine one, i.e.,
	 \[
	g(g^{-1}h)=g^{-1}(gh)=h, \quad \forall g,h\in \mathbb{O},
	 \]
which is stronger than the standard property (\ref{inverse}) for non-associative algebra.

	Actually, the set $\Oc^\ti$ of invertible octonions is a smooth Moufang loop under octonion multiplication.

	One may represent an octonion as a pair of quaternions $\mathcal Q$, then multiplication can be defined by
\[
(a,b)\cdot(c,d)=(ac-d^*b,da+bc^*), \quad \mbox{for} \; a,b,c,d\in \mathcal Q\,,
\]
where the involution, addition and multiplication are those in quaternions. In this case the inverse of $(a,b)$ is given by $$(a,b)^{-1}=\displaystyle\frac{(a,b)^*}{\parallel a \parallel ^2+\parallel b \parallel ^2}\,,$$ where $(a,b)^*=(a,-b^*)$ and the norm is in quaternions.

\begin{example}\label{ex1}	
	The set of all automorphisms of the algebra $\mathbb{O}$, that is the set of invertible linear transformations $A\in Aut (\mathbb{O})$, forms a Lie group called $G_2$ which is the smallest of the exceptional Lie groups. We will show that the semidirect product $\mathbb{O}^\ti \ltimes G_2 $ is an inverse loop under the multiplication
	\[
	(g,A)\bullet(h,B)=(g A(h), A\circ B),
	\]
	with identity $(1,\Id)$ and inverse $(g,A)^{-1}=(A^{-1}(g^{-1}),A^{-1})$. The thing which needs to be checked is the following inverse property,
	\[
	\begin{array}{rcl}
	(g,A)^{-1}\bullet \left( (g,A)\bullet (h,B)\right) &=&(A^{-1}(g^{-1}),A^{-1})\bullet (g\cdot A(h), A\circ B)\\[3pt]
	&=& (A^{-1}(g^{-1})\cdot A^{-1}(g\cdot A(h)),B)=(h,B)\,.
	\end{array}
	\]
	Here we use the fact that $A^{-1}(g^{-1}\cdot g)=A^{-1}(g^{-1})\cdot A^{-1}(g)=1$. Similarly,
\[
 \left( (g,A)\bullet (h,B)\right)\bullet (h,B)^{-1}\,.
 \]
Of course, because $\Oc^\ti$ are not associative the above smooth loop is not a Lie group.
	
\end{example}
\begin{example}\label{ex2}
The manifold of unitary octonions
\[
{\mathbb{O}}_1	=\{a\in \mathbb{O}, \; \|a\|=1\}
\]
is closed under the octonion multiplication and therefore forms a Moufang loop. The manifold $\mathbb{O}_1$ is diffeomorphic to seven-sphere $S_7$, the only paralellizable sphere which does not carry a Lie group structure. To find the tangent algebra of $\mathbb{O}_1$, consider the tangent space
\[
\mathbf o_1	=\sT_{e_0}\mathbb{O}_1=span\{e_1,...,e_7\}
\]
to $\Oc_1$ inside the vector space $\Oc$.
Then, the tangent bundle $\sT \mathbb{O}_1$ is given by the left (or right) prolongation of imaginary octonions, that is $\sT \mathbb{O}_1 =span \{\lvec e_1,...,\lvec e_7\}$, where
\[
\lvec e_i(a)=D_{e_0}(l_a)(e_i)=ae_i\in \sT_a\mathbb{O}_1,\quad a\in \mathbb{O}_1.	
\]
Similarly
\[
\rvec e_i(a)=D_{e_0}(r_a)(e_i)=e_ia\in \sT_a\mathbb{O}_1,\quad a\in \mathbb{O}_1.	
\]

We have
\[
\begin{array}{rcl}
[\lvec {e_i},\lvec {e_j}]{(a)}&=&\displaystyle\frac{d}{dt}|_{t=0}\lvec e_j(a+tae_i)-\frac{d}{dt}|_{t=0}\lvec e_i(a+tae_j)\\[6pt]
&=&\displaystyle\frac{d}{dt}|_{t=0}(a+tae_i)e_j-\frac{d}{dt}|_{t=0}(a+tae_j)e_i=(ae_i)e_j-(ae_j)e_i\\[6pt]
&=&\displaystyle\lvec{a^{-1}((ae_i)e_j-(ae_j)e_i)}(a)\,.
\end{array}
\]	
In particular $[\lvec {e_i},\lvec {e_j}]{(e_0)}=2e_ie_j$.
Thus, $(\mathbf o_1,[e_i,e_j]=2e_ie_j)$ is the skew-algebra (Mal'cev algebra) corresponding to the smooth loop $\mathbb{O}_1$. The corresponding Leibniz structure is
$$\zL=\sum_{i,j=1}^7e_ie_j\pa_{e_i}\we\pa_{e_j}\,.$$
This agrees with the Campbell-Hausdorff formula. The exponential map $\exp:\o_1\to\Oc_1$ is
\be\label{ex}\exp(e)=\cos(\nm{e})+\sin(\nm{e})e/\nm{e}\,.\ee
Hence,
$$\exp(te_i)\exp(te_j)=\cos^2(t)+\sin(t)\cos(t)(e_i+e_j)+\sin^2(t)e_ie_j=\exp(e)\,.$$
If $i\ne j$, then, in view of (\ref{ex}),
	\beas
 e&=&\frac{\arcsin(\sqrt{1-\cos^4(t)})}{\sqrt{1-\cos^4(t)}}\left(\sin(t)\cos(t)(e_i+e_j)+\sin^2(t)e_ie_j\right)\\
&=&(1+o(t))\left(t(e_i+e_j)+t^2e_ie_j+o(t^2)\right)=\left(t(e_i+e_j)+t^2e_ie_j+o(t^2)\right)\,.
\eeas
Since the Campbell-Hausdorff formula reads
$$te_i*tej=t(e_i+e_j)+(t^2/2)[e_1,e_2]+o(t^2)\,,$$
we get again
$$[e_1,e_2]=2e_ie_j\,.$$

It is also easily seen that $\lvec e_1,\dots,\lvec e_7$ do not form a Lie algebra over $\R$ (i.e., $\Oc$ is not a Lie group).

Similarly we obtain
\[
[\rvec {e_i},\rvec {e_j}]{(a)}=e_j(e_ia)-e_i(e_ja)
\]	
so that $[e_i,e_j]_l=-[e_i,e_j]_r=2e_ie_j$. Note that
\[
[\rvec {e_i},\lvec {e_j}]{(a)}=e_j(ae_i)-(e_ja)e_i
\]	
which is $0$ at $a=e_0$, but generally not $0$ (the left prolongations do not commute with the right prolongations) due to non-associativity of $\Oc$.

\end{example}
\begin{example}\label{ex3}
Consider the cotangent bundle $\sT ^*\mathbb{O}_1$ of unit octonions. Due to the scalar product (\ref{inner}), $\sT ^*\mathbb{O}_1= \sT \mathbb{O}_1$ as vector bundles. According to the general rule for the Lie groupoid $\sT^*G\rightrightarrows \mathfrak g^*$ in case of a Lie group $G$, we define the source and target projections $ \alpha,  \beta:\sT ^*G\to \mathbf o^*_1$ where $\mathbf o^*_1$ is the dual of $\mathbf o_1$:
\[
\left\langle  \beta (\mu_g),X\right\rangle = \left\langle \mu_g ,gX\right\rangle\,,\quad\left\langle  \alpha (\nu_h),X\right\rangle=\left\langle \nu_h  ,Xh\right\rangle\,.
\]
Here we use a self-explaining notation $\lvec X_g=gX$, $\rvec X_h=Xh$ and we interpret the tangent and cotangent vectors to $\Oc_1$ as octonions. In this sense the above pairings can be understood as the scalar products.

Two elements $\mu_g\in \sT ^*_g\mathbb{O}_1$ and $\nu_h\in \sT ^*_h\mathbb{O}_1$ are \emph{composable}, i.e., $ \beta(\mu_g)= \alpha (\nu_h)$, if
\[
\left\langle  \beta (\mu_g),X\right\rangle =\left\langle  \alpha (\nu_h),X\right\rangle
\quad\Leftrightarrow\quad
\left\langle \mu_g ,gX\right\rangle=\left\langle \nu_h  ,Xh\right\rangle \quad \text{for all}\quad X\in\mathbf o_1.
\]

The above holds if there exists an element $\sigma\in \mathbf o^*_1=\o_1$ such that

\[
\mu_g=D_g^*(l_{g^{-1}})(\sigma)=g\sigma \quad \mbox{and} \quad \nu_h=D_h^*(r_{h^{-1}})(\sigma)=\sigma h.
\]

Let us try to define the product in $\sT^*\Oc_1$ like it is done for Lie groups. The multiplication $\zm_g\bullet\zn_h\in\sT^*_{gh}\Oc_1$
is then defined by the equation
\be\label{e4} \langle\zm_g,X_g\rangle+\langle\zn_h,Y_h\rangle=\langle\zm_g\bullet\zn_h, X_g\bullet Y_h\rangle\,,
\ee
where $ X_g\bullet Y_h$ is the multiplication in the tangent loop as described in Theorem \ref{t2}.
Denote $\zvy_{gh}=\zm_g\bullet\zn_h$. Using Theorem \ref{t2} we can rewrite (\ref{e4}) as
\be\label{e5}
\langle\zm_g,X_g\rangle+\langle\zn_h,Y_h\rangle=\langle\zvy_{gh},X_gh+gY_h\rangle.
\ee
 Since $X_g,Y_h$ are arbitrary, so we get out of (\ref{e5}) a system of equations
$$\langle\zm_g,X_g\rangle=\langle\zvy_{gh},X_gh\rangle\,,\quad \langle\zn_h,Y_h\rangle=\langle\zvy_{gh},gY_h\rangle\,.
$$
Using the fact that the scalar product on $\Oc_1$ is invariant with respect to the multiplication, we get in turn

$$\langle\zm_g,X_g\rangle=\langle\zvy_{gh}h^{-1},X_g\rangle\,,\quad \langle\zn_h,Y_h\rangle=\langle g^{-1}\zvy_{gh},Y_h\rangle\,.$$
Hence, $\zm_g=g\zs=\zvy_{gh}h^{-1}$ and $\zn_h=\zs h=g^{-1}\zvy_{gh}$, so that
$$\zs=g^{-1}(\zvy_{gh}h^{-1})=(g^{-1}\zvy_{gh})h^{-1}\,.$$
In this way we get the commutativity of the right translation by $h^{-1}$ and left translation by $g^{-1}$  when acting on
$\zvy_{gh}$ which can be taken arbitrary octonion orthogonal to $gh$. This is not satisfied in octonions as $\Oc$ is non-associative. This implies that the standard way of defining the multiplication in $\sT^*\Oc_1$ does not give a well-defined product.
\end{example}
\section{Mechanics on octonions}
The algebra of octonions $\Oc$ is spanned by $\{ e_0=1\}$-the unit and 7 additional unitary elements
$\{ e_1,e_2,\dots, e_7\}$, $e_i^2=-1$, $e_ie_j=-e_je_i$ for $i\ne j$. The algebra is non-commutative and non-associative (e.g. $(e_1e_4)e_7=e_5e_7=e_2\ne e_1(e_4e_7)=e_1e_3=-e_2)$. In this section we will construct the discrete mechanics on the manifold of unit octonions
\[
{\mathbb{O}}_1	=\{a\in \mathbb{O}, \; \|a\|=1\},
\]
which is an inverse smooth loop under the octonion multiplication.
Let $L: \mathbb{O}_1\to \mathbb R$ be a Lagrangian function, then the discrete Euler-Lagrange equations read as recurrence equation
	\[
	\lvec e_i(L)(a_n)=\rvec e_{i}(L)(a_{n+1})\,,
	\]
	where $\lvec e_i(a)=D_{e_0}(l_a)(e_i)=ae_i $ and $\rvec e_i(a)=D_{e_0}(r_a)(e_i)=e_ia$ are the left and the right prolongation by the element $a\in \mathbb{O}_1$. A solution for those equations is a sequence of elements $\mathbb{O}_1$.
	
	If we take the Lagrangian as a linear function, for instance take $L=e^1=\left\langle e_1,\cdot\right\rangle $ defined by the inner product (\ref{inner}), then
	\[
	\lvec e_i(L)(a_n)=(a_n e_i)(L)(a_n)=\left\langle e_1,a_ne_i\right\rangle .
	\]
	The right-hand side of the above relation is obtained by taking the integral curve $\gamma(t) =a_n+ta_ne_i$ for the tangent vector $a_ne_i$ and then we have		
	\[
	\frac{d}{dt}|_{t=0} L(a_n+ta_ne_i)=\frac{d}{dt}|_{t=0}\left\langle e_1,a_n+ta_ne_i\right\rangle .
	\]
	Note that we interpret the tangent vector at the points of $\mathbb{O}_1$ as an element of octonions.
	
	Therefore, by the definition the Euler-Lagrange equations are
	\[
	\left\langle e_1,a_ne_i-e_ia_{n+1}\right\rangle =0,\quad \mbox{for}\quad i=1,...,7.
	\]			
	Every element $a_n\in \mathbb{O}_1$ can be written as $a_n=\alpha_n^0+\alpha^k_ne_k$ such that $\sum _{s=0} ^7|\alpha_n^s|^2=1$, so the above equations turn to
		\begin{equation}\label{ELE}
\sum_{k=1}^7	\left\langle e_1,\alpha_n^0e_i-\alpha^0_{n+1}e_i+(\alpha_n^k+\alpha_{n+1}^k)e_ke_i\right\rangle=0,\quad\mbox{for}\quad i=1,...,7.
	\end{equation}  	
	Now, if $i=1$, since $\left\langle e_1,e_1\right\rangle =1$ and $\left\langle e_1,e_ke_1\right\rangle =0$ for $k\ne 0$, we get $\alpha_n^0-\alpha_{n+1}^0=0$.
	
	If $i>1$, the two first expressions of (\ref{ELE}) are zero because $\left\langle e_1,e_i\right\rangle =0$ for $i\ne 1$ and thus we left by the third expression, that is
	\[
\sum_{k=1}^7\left\langle e_1,(\alpha_n^k+\alpha_{n+1}^k)e_ke_i\right\rangle=0,\quad\mbox{for}\quad i=1,...,7\,.	\]
	But for each $k$, there is some $i$ such that $e_ke_i=\pm e_1$ and all $i's$ are used.
 Consequently, we get the Euler-Lagrange equations
	\[
	\alpha_n^0-\alpha_{n+1}^0=0, \quad \alpha_n^k+\alpha_{n+1}^k=0,\quad k=1,...,7.
	\]
	It is obvious from the equations that the solution of Euler-Lagrange equations are just the conjugate pairs in $\mathbb{O}_1$.
	
Next step is to check whether  the Lagrangian $L$ is regular or hyperregular. So, we would need to find the Legendre maps associated with $L$. Consider the tangent skew-algebra $\mathbf o_1$ and the its dual $\mathbf o^*_1$ with the basis $\{e^1,...,e^7\}$. The corresponding Legendre maps $\mathbb F^+ L,\mathbb F^- L:\mathbb{O}_1\to \mathbf o^*_1$ are
\[
\mathbb F^+ L(a)=\sum_{i=1}^7 \lvec e_i(a)(L) e^i,\quad \mathbb F^- L(a)=\sum_{i=1}^7 \rvec e_i(a)(L) e^i,\quad a\in \mathbb{O}_1.
\]
Let us remark that there is no hyperregular Lagrangian on unit octonions $\mathbb{O}_1$, because the Legendre maps are $\mathbb F^+ L,\mathbb F^- L:S_7\to \mathbb R^7$ which cannot be diffeomorphisms. Thus we can only find Lagrangians which are (locally) regular.

Consider the linear Lagrangian $L=\left\langle e_1,\cdot\right\rangle=e^1 $. We have $\rvec e_i(a)(e^1)=\left\langle e_1,e_ia\right\rangle $ and corresponding Legendre map
\[
\mathbb F^-L(a) =\sum_{i=1}^7 \rvec e_i(e^1)(a) e^i=\sum_{i=1}^7\left\langle e_1,e_ia\right\rangle e^i,\quad a\in \mathbb{O}_1.
\]	
The Lagrangian $L=e^1$ is not regular at $e_0$ because
	\[
\mathbb F^-L(e_0) =\sum_{i=1}^7\left\langle e_1,e_i\right\rangle e^i=e^1
\]	
and
\[
D_{e_0}(\mathbb F^-L)(e_1)=\displaystyle \frac{d}{dt}_{|_{t=0}}\mathbb F^-L(e_0+te_1)=\sum_{i=1}^7\left\langle e_1,e_ie_1\right\rangle e^1 =0.
\]

	If we take the Lagrangian $L= \displaystyle \frac{(e^1)^2}{2}$, then
\[
\rvec e_i(a)(L)=e^1(a)\rvec e_i(a)(e^1)=\left\langle e_1,a\right\rangle\left\langle e_1,e_i a\right\rangle
\]

and $\mathbb F^-L(a)=\sum_{i=1}^7\left\langle e_1,a\right\rangle\left\langle e_1,e_i a\right\rangle e^i$. We have $\mathbb F^-L(e_s) =0$ for every $s=0,...,7$ and $L$ is not regular in a neighbourhood of $e_0$.  Indeed,

\[
D_{e_0}(\mathbb F^-L)(e_s)=\displaystyle \frac{d}{dt}_{|_{t=0}}\mathbb F^-L(e_0+te_s)=\sum_{i=1}^7\left\langle e_1,e_s\rangle \langle e_1,e_i\right\rangle e^i =0,\quad \mbox{for}\quad s\neq 1.
\]

The discrete Euler-Lagrange equation is
$$\left\langle e_1,a_n\right\rangle\left\langle e_1,e_i a_n\right\rangle=\left\langle e_1,a_{n+1}\right\rangle\left\langle e_1, a_{n+1}e_i\right\rangle\,,\quad i=1,\dots,7\,.$$

If we write $a_n=\sum_{s=0}^7\za^s_ne_s$ and $a_{n+1}=\sum_{s=0}^7\za^s_{n+1}e_s$, then this reduces to the quadratic recurrence equation

$$\za_n^1\za_n^0=\za_{n+1}^1\za_{n+1}^0,\quad \mbox{and}\quad  \za^1_n\za^j_n=-\za^1_{n+1}\za^j_{n+1}\,,\quad \mbox{for}\quad j=2,\dots 7\,.$$

Now, we take the Lagrangian
\be\label{lag}L=\sum_{k=1}^7 \displaystyle m_k\frac{(e^k)^2}{2}\,,
\ee
where $m_k> 0$, as the `total kinetic energy' of the system. Then
\[
\rvec e_i(a)(L)=\sum_{k=1}^7m_ke^k(a)\rvec e_i(a)(e^k)=\sum_{k=1}^7\left\langle m_ke_k,a\right\rangle\left\langle e_k,e_i a\right\rangle,
\]
and

\[
\lvec e_i(a)(L)=\sum_{k=1}^7m_ke^k(a)\lvec e_i(a)(e^k)=\sum_{k=1}^7\left\langle m_ke_k,a\right\rangle\left\langle e_k,ae_i \right\rangle,
\]
so that

\begin{equation}\label{FL}
\mathbb F^-L(a)=\sum_{i,k=1}^7\left\langle m_ke_k,a\right\rangle\left\langle e_k,e_i a\right\rangle e^i.
\end{equation}
and
\begin{equation}\label{F+L}
\mathbb F^+L(a)=\sum_{i,k=1}^7\left\langle m_ke_k,a\right\rangle\left\langle e_k,ae_i \right\rangle e^i.
\end{equation}

The discrete Euler-Lagrange equation is	
$$\sum_{k=1}^7\left\langle m_k e_k,a_n\right\rangle\left\langle e_k,e_i a_n\right\rangle=\sum_{k=1}^7\left\langle m_ke_k,a_{n+1}\right\rangle\left\langle e_k, a_{n+1}e_i\right\rangle\,,\quad i=1,\dots,7\,.
$$

If we write $a_n=\sum_{s=0}^7\za^s_ne_s$ and $a_{n+1}=\sum_{s=0}^7\za^s_{n+1}e_s$, then this reduces to the quadratic recurrence equation
$$\za_n^i\za_n^0=\za_{n+1}^i \za_{n+1}^0,\quad i=1,...,7$$
which does not depend on $m_k$.
Since $\sum_{i=0}^7\left(\za_{n}^i\right)^2=1$, we have
\be\label{rec}\sum_{i=1}^7\left(\za_{n}^i\za_n^0\right)^2=(\za_n^0)^2\left(1-(\za_n^0)^2\right)=(\za_{n+1}^0)^2\left(1-(\za_{n+1}^0)^2\right)
\ee
which have non-trivial different solutions giving rise to non-trivial solutions of the Euler-Lagrange equation.
If, however, $\za_n^0$ and $\za_{n+1}^0$ are close to 1, thus $a_n$ and $a_{n+1}$ are close to $e_0$, then
$\za_n^0=\za_{n+1}^0$, since the function $f(x)=x^2(1-x^2)$ is monotonic on the interval $[1/\sqrt{2},1]$, so we get  only trivial solutions $a_n=a_{n+1}$.
\begin{theorem}\label{tel}
The discrete Euler-Lagrange equation for the Lagrangian (\ref{lag}) on the smooth loop $\Oc_1$ admits in a neighbourhood of $e_0$ only trivial solutions.
\end{theorem}

\noindent We easily see see that $\mathbb F^-L(e_s) =0$ for all $s\in\{0,...,7\}$, hence, changing the base of imaginary octonions, we infer that $\mathbb F^-L(a) =0$ for any imaginary octonion, that supports once more the fact that the Lagrangian is not hyperregular. We have also
\[
D_{e_0}\mathbb F^-L(e_s)=\sum_{i,k=1}^7\left\langle m_ke_k,e_s\right\rangle\left\langle e_k,e_i\right\rangle e^i=m_se^s=D_{e_0}\mathbb F^+L(e_s)\,.
\]
	
Under identification $\sT_{e_0}\Oc_1=\o_1=\o_1^*$ the differential of $\mathbb F^-L:\Oc_1\to\o_1$ at $e_0$, and similarly for $\mathbb F^+L:\Oc_1\to\o_1$, can be identified with the diagonal automorphism on $\o_1$ for which
$e_s\mapsto m_se_s$. In particular $\mathbb F^-L$ and $\mathbb F^+L$ are regular in a neighbourhood of $e_0$.
According to Theorem \ref{tel} and (\ref{df}),  $\mathbb F^-L=\mathbb F^+L$ in a neighbourhood of $e_0$
Actually, $\mathbb F^-L(a)=\mathbb F^+L(a)$ for all $a\in\Oc_1$, because they are quadratic (thus analytic) in $a$.
Hence, the discrete Lagrangian evolution operator in a neighbourhood $U$ of $e_0$ is
$$\zg_L=(\mathbb F^- L)^{-1}\circ \mathbb F^+ L=\Id_{U}
$$
and the local discrete Hamiltonian operator in a neighbourhood $V$ of $0$ in $\o_1^*$ reads
$$\tilde\zg_L=\mathbb F^+ L\circ (\mathbb F^- L)^{-1}=\Id_{V}\,.
$$
Note however that there are nontrivial solutions lying outside the neighbourhood of $e_0$. For instance,
$(a_n,a_{n+1})=\left((0,a'_n),(0,a'_{n+1})\right)$, where $a'_n$ and $a'_{n+1}$ represent arbitrary imaginary and unitary octonions. We can also take $A=\za^0_n\ne \za^0_{n+1}=B$, $|A|<|B|<1$, which are different solutions of (\ref{rec}). Note that in this case $B$ is not close to 1, as $|A|<1/\sqrt{2}$.
Then, for any $a_n=(A, a'_n)$, where $a'_n$ represents an imaginary octonion of length $\sqrt{1-A^2}$, the pair
$\left((A,a'_n),(B, u a'_n)\right)$, with $u^2=(1-B^2)/(1-A^2)$, is a solution of the Euler-Lagrange equation.

\section*{Acknowledgements}
J.~Grabowski acknowledges that his  research was funded by the  Polish National Science Centre grant HARMONIA under the contract number 2016/22/M/ST1/00542.

\section{Concluding remarks} We have shown how the discrete Lagrangian and Hamiltonian formalism on Lie groupoids can be extended to non-associative objects like smooth inverse loops. The working example was the inverse loop of unitary octonions $\Oc_1$ on which we defined a `linear' Lagrangian for which solutions of the discrete Euler-Lagrange equations are pairs of conjugate octonions. We also gave an example of a regular Lagrangian and discuss the corresponding Euler-Lagrange equations.

Since the discrete mechanics has its version on Lie groupoids (with fundamental for our paper works \cite{MMM,weinstein}), a natural question is to find the non-associative generalizations of Lie groupoids and to construct discrete mechanics on them. The first ideas of such objects, \emph{smooth loopoids}, can be found in \cite{Gr, Kin}. We will discuss this problem in a separate paper.


\small{\vskip1cm

\noindent Janusz GRABOWSKI\\ Institute of
Mathematics\\  Polish Academy of Sciences\\ \'Sniadeckich 8, 00-656 Warszawa, Poland
\\Email: jagrab@impan.pl \\

\noindent Zohreh RAVANPAK\\ Institute of
Mathematics\\  Polish Academy of Sciences\\ \'Sniadeckich 8, 00-656 Warszawa, Poland
\\Email: zravanpak@impan.pl \\


\begin{thebibliography}{99}
		
		\bibitem{Bae}  R. Baer, {\it Nets and groups}, {Trans. Amer. Math. Soc.} {\bf 46} (1939), 110â€“141.
		
		\bibitem{Bel} V. D. Belousov, {\it Foundations of the theory of quasigroups and loops}, Nauka,
		Moscow, 1967 (in Russian).
		
		
		
		
		
		
		
		\bibitem{Br1} R.~H.~Bruck, {\it What is a loop}, in "Studies in modern algebra", 59--99, Studies in Mathematics vol. 2, New Jersey Prentice-Hall 1963.
		
		\bibitem{FZ} Y. N. Fedorov, D.V. Zenkov, {\it Discrete nonholonomic LL systems on Lie groups},
		{Nonlinearity} {\bf 18} (2005), 2211--2241.
		
		\bibitem{Fo} T. Foguel, {\it Groups, transversals, and loops}, Loops'99 (Prague), Comment. Math. Univ. Carolin.
		{\bf 41}  (2000), 261--269.
		
		\bibitem{GG} K. Grabowska and J. Grabowski, {\it Variational calculus with constraints on general algebroids},
		{J. Phys. A} {\bf 41} (2008), 175204 (25pp).
		
		\bibitem{GG1} K. Grabowska and J. Grabowski, {\it Dirac Algebroids in Lagrangian and Hamiltonian Mechanics},
		J. Geom. Phys. {\bf 61} (2011), 2233--2253.
		
		\bibitem{GGU} K.~Grabowska, J.~Grabowski and P.~Urba\'nski,
		{\it Geometrical Mechanics on algebroids}, {Int. J. Geom. Meth. Mod. Phys.} {\bf 3} (2006), 559-575.
		
		\bibitem{GGKM} K.~Grabowska, J.~Grabowski, M.~Ku\'s and G.~Marmo, {\it Lie groupoids in information geometry}, J. Phys. A {\bf 52} (2019), 505202 (22pp).
		
		\bibitem{Gr}	J. Grabowski,
		{\it An introduction to loopoids}, Comment. Math. Univ. Carolin. {\bf 57} (2016), 515-526
		
		\bibitem{GJ} J.~Grabowski and M.~J\'o\'zwikowski,
		{\it Pontryagin Maximum Principle on almost Lie algebroids}, {SIAM
			J. Control Optim.} {\bf 49} (2011), 1306--1357.
		
		\bibitem{GLMM} J.~Grabowski, M.~de Leon, J.C.~Marrero, D.~Martin de Diego,
		{\it Nonholonomic Constraints: a New Viewpoint}, {J. Math. Phys.} {\bf 50} (2009), 013520 (17pp).
		
		\bibitem{GU1} J.~Grabowski and P.~Urba\'nski, {\it Lie algebroids and
			Poisson-Nijenhis structures}, {Rep. Math. Phys.} {\bf 40}
		(1997), 195--208.
		
		\bibitem{GU2} J.~Grabowski and P.~Urba\'nski, {\it Algebroids -- general differential calculi on vector bundles}, {J. Geom. Phys.} {\bf 31} (1999), 111-141.
		
		
		\bibitem{IMMM}
		D. Iglesias, J. C. Marrero, D. Mart\'{\i}n de Diego, E.
		Mart{\'\i}nez, {\it Discrete nonholonomic Lagrangian systems on Lie
			groupoids}, {J. Nonlinear Sci.}  {\bf 18}  (2008),  221--276.
		
		\bibitem{IMMP}
		D. Iglesias, J. C. Marrero, D. Mart\'{\i}n de Diego, E.
		Padr\'on, {\it Discrete nonholonomic in implicit form}, Discrete Contin. Dyn. Syst.  {\bf 33}  (2013),  1117--1135.
		
		\bibitem{IMMS} D. Iglesias, J. C. Marrero, D. Mart\'{\i}n de Diego, D. Sosa, {\it Singular Lagrangian systems and variational constrained mechanics on Lie algebroids}, {Dynamical Systems}, {\bf 23} (2008), 351--397.
		
		\bibitem{JM}
		J.M. P\'erez-Izquierdo, {\it An Envelope for Bol Algebras}, {J. Algebra} {\bf 284} (2005), 480–493.
		
		
		
		\bibitem{Kin} M.~Kinyon, {\it The coquecigrue of a Leibniz algebra}, {preprint}, 2003.
		
		\bibitem{KW} M.~Kinyon, A.~Weinstein, {\it Leibniz algebras, Courant algebroids, and multiplications on reductive homogeneous spaces}, Amer. J. Math.  {\bf 123}  (2001), 525--550.
		
		\bibitem{kuz} E. N. Kuzʹmin,
		{\it The connection between Malʹcev algebras and analytic Moufang loops.} (Russian),
		Algebra i Logika {\bf 10} (1971), 3--22.
		
		
	\bibitem{LMM}
M.~de~Le\'on, J.~C.~Marrero \& E.~Mart{\'\i}nez,
\newblock{Lagrangian submanifolds and dynamics on Lie algebroids,}
\newblock{ \emph{J. Phys. A} \textbf{38} (2005), no. 24, 241--308.}

\bibitem{Li}  P.~Libermann,
\newblock{ Lie algebroids
and mechanics,}
\newblock{ {\em Archivum Mathematicum} {\bf 32}, (1996), 147--162.}	
	
		
		\bibitem{Ma} K. Mackenzie,
		{\it General Theory of Lie Groupoids and Lie Algebroids}, {London Mathematical Society Lecture Note Series}, {\bf 213}, Cambridge University Press, Cambridge, 2005.
		
		\bibitem{mal} A. I. Mal'cev, {\it Analytic loops}, Mat. Sb. N.S. (in Russian), {\bf 36}(78) (1955), 569--576
		
		
		
	
		
		\bibitem{MMM}
		J. C. Marrero, D. Mart\'{\i}n de Diego, E. Mart\'{\i}nez, {\it Discrete Lagrangian and Hamiltonian Mechanics on Lie
			groupoids}, {\it Nonlinearity} {\bf 19} (2006), 1313--1348. Corrigendum: {\it Nonlinearity} {\bf 19} (2006), 3003--3004.
	
		
		\bibitem{MP}		J. E. Marsden, S. Pekarsky and S. Shkoller (1999b),{\it Symmetry reduction of discrete
		Lagrangian mechanics on Lie groups}, {J. Geom. Phys.} {\bf 36}, 140–151.
		
			\bibitem{MR}
		J.E. Marsden and T.S. Ratiu, {\it Introduction to Mechanics and Symmetry, SpringerVerlag}, (1994). {Second Edition, 1999}.
		
		
		\bibitem{MW} J.E. Marsden, M. West, {\it Discrete mechanics and variational integrators}, {Acta Numerica} (2001), 357--514.
		
	
		
		
		
		\bibitem{MMS}
		J. C. Marrero, D. Mart\'{\i}n de Diego, A. Stern, {\it Symplectic groupoids and discrete constrained Lagrangian mechanics}, Discrete Contin. Dyn. Syst.  {\bf 35}  (2015), 367--397.
		
		\bibitem{Martinez}
		E. Mart\'{\i}nez,
		{\it Lagrangian Mechanics on Lie algebroids,}
		{Acta Appl. Math.}, \textbf{67} (2001), 295--320.
		
		
		\bibitem{MV}
		Moser, J., and Veselov, A.P.,
		Discrete versions of some classical integrable systems and factorization of matrix polynomials, {\it Comm. Math. Phys.} {\bf 139} (1991), 217-243.
		
		
		
		\bibitem{nagy} P. T. Nagy,
		{\it Moufang loops and Malcev algebras},
		Seminar Sophus Lie {\bf 3} (1992), 65--68.
		
		\bibitem{Pfl} H. O. Pflugfelder: {\it Quasigroups and Loops: Introduction}, Berlin, Heldermann
		Verlag, 1990.
		
		
		\bibitem{Sab} L. V. Sabinin, {\it Smooth Quasigroups and Loops}, Kluwer Academic Press, 1999.
		
		\bibitem{Sab1} L. V. Sabinin, {\it Smooth quasigroups and loops: forty-five years of incredible growth}, {Comment. Math. Univ. Carolin.} {\bf 41} (2000), 377--400.
		
		
		\bibitem{Sc}
		R. D. Schafer,
		{\it An Introduction to Nonassociative Algebras},
		{Academic
			Pres, New York,} \textbf{} (1966).
		
		\bibitem{Smi} J. D. H. Smith, {\it Loops and quasigroups: Aspects of current
			work and prospects for the future}, {Comment. Math. Univ. Carolin.} {\bf 41} (2000), 415--427.
		
		\bibitem{Stern}
		A. Stern, {\it Discrete {H}amilton--{P}ontryagin mechanics and generating functions on {L}ie groupoids}, {J. Symplectic Geom.} {\bf 8} (2010), 225--238.
		
		
		\bibitem{Tul1} W. Tulczyjew: Les sous-vari\'{e}t\'{e}s lagrangiennes et la
		dynamique hamiltonienne, {\it C.R. Acad. Sci. Paris} {\bf 283} (1976), 15-18.
		
		\bibitem{Tul2} W. Tulczyjew: Les sous-vari\'{e}t\'{e}s lagrangiennes et la
		dynamique lagrangienne, {\it C.R. Acad. Sci. Paris} {\bf 283} (1976), 675-678.
		
		\bibitem{We} A.~Weinstein, {\it Symplectic groupoids and Poisson manifolds},
		{\it Bull. Amer. Math. Soc.} {\bf 16} (1987), 101--104.
		
		
		
		
		\bibitem{weinstein}
		A. Weinstein, {\it Lagrangian Mechanics and groupoids}, {Fields Inst. Comm.} {\bf 7} (1996), 207--231.
		
		
	\end{thebibliography}
\end{document}